\documentclass{article}
\usepackage{amsmath,amsthm,amssymb}

\newtheorem{theorem}{Theorem}[section]
\newtheorem{lemma}[theorem]{Lemma}
\newtheorem{remark}[theorem]{Remark}
\newtheorem{definition}[theorem]{Definition}
\newtheorem{corollary}[theorem]{Corollary}

\newtheorem{example}[theorem]{Example}

\numberwithin{equation}{section}
\numberwithin{table}{section}

\def\Z{{\mathbb Z}}\def\F{{\mathbb F}}

\def\T{{\varphi}} \def\P{{\rho}}
\def\CRT{\,\mathop{=\kern-3pt=}\limits^{\mbox{\rm\tiny CRT}}\,}

\begin{document}

\title{Isometrically Self-dual Cyclic Codes}

\insert\footins{\footnotesize
{\it Email addresses}:  yfan@mail.ccnu.edu.cn (Y. Fan)}

\author{Yun Fan\quad and\quad Liang Zhang\\
\small Dept of Mathematics,
\small  Central China Normal University, Wuhan 430079, China}


\maketitle

\begin{abstract}
General isometries of cyclic codes, including multipliers and translations,
are introduced; and isometrically self-dual cyclic codes are defined.
In terms of Type-I duadic splittings given by multipliers and translations,
a necessary and sufficient condition for the existence
of isometrically self-dual cyclic codes is obtained.
A program to construct isometrically self-dual cyclic codes is provided,
and illustrated by several examples. In particular,
a class of isometrically self-dual MDS cyclic codes,
which are alternant codes from a class of generalized Reed-Solomon codes,
is presented.


\end{abstract}

\section{Introduction}

Self-dual codes and cyclic codes 
are two important and long-time topics
which are interesting in both theoretical perspective
and technological practice.

Let $\F_q$ be a finite field with $q$ elenments,
where $q$ is a power of a prime.
Cyclic codes of length $n$ over $\F_q$
can be described as ideals of the quotient algebra
$\F_q[X]/\langle X^n-1\rangle$ of the polynomial
algebra $\F_q[X]$ with respect to the ideal $\langle X^n-1\rangle$
generated by $X^n-1$.
If the greatest common divisor $\gcd(q,n)=1$,
then $\F_q[X]/\langle X^n-1\rangle$ is a semisimple algebra.
Semisimple cyclic codes are studied and applied extensively.
However, it is a pity that there are no (euclidean)
self-dual semisimple cyclic codes.
An alternative research is the study on
{\em duadic} cyclic codes, 
which are initiated by Leon, Masley and Pless \cite{LMP},
and developed extensively from various directions, e.g., see 
\cite{P, S, BP, R, WZ, DLX, DP, W, LX, MW, SBDR, AKS, HK, FZ12, JLS, FZ16-a}.

In the semisimple case, the irreducible factors in $\F_q[X]$ of ${X^n-1}$
correspond one-to-one to the {\em $q$-cyclotomic cosets}
({\em $q$-cosets} in short) of the residue ring $\Z_n$
of the integer ring $\Z$ modulo $n$.
By $\Z_n^*$ we denote the subset of the elements $s\in\Z_n$
such that $\gcd(s,n)=1$.
For any $s\in\Z_n^*$,
we have a permutation $\mu_s$ of $\Z_n$
such that (see \cite[p.138]{HP}):
\begin{equation}\label{multiplier}
 \mu_s(i)=si\!\!\pmod n,~~~~~~~ \forall~i\in\Z_n.
\end{equation}
The permutation $\mu_s$ is called a {\em multiplier}.
Correspondingly,
we have an isomorphism $\hat\mu_s$ of the algebra
$\F_q[X]/\langle X^n-1\rangle$ such that (see \cite[Eqn (4.4)]{HP}):
\begin{equation}\label{isometry}
 \hat\mu_s\big(a(X)\big)=a(X^{s})\!\!\pmod{X^n-1}, ~~~~~
 \forall~ a(X)\in \F_q[X]/\langle X^n-1\rangle.
\end{equation}
The isomorphism $\hat\mu_s$ preserves the Hamming weight obviously
(the isomorphisms which preserve Hamming weight
are called {\em isometries}).
If there is a subset $P\subseteq\Z_n$ which is a union of some $q$-cosets
such that $\Z_n=\{0\}\bigcup P\bigcup\mu_s(P)$ is a partition,
then $P$ and $\mu_s(P)$ are called an even-like duadic splitting of $\Z_n$.
In that case, $\F_q[X]/\langle X^n-1\rangle=C+\hat\mu_s(C)$
and $C\bigcap\hat\mu_s(C)$ is the cyclic code generated by
the vector $(1,\cdots,1)$, where $C$ is the cyclic code with
defining set~$P$.
Please refer to \cite[Ch. 6]{HP} for details.

Constacyclic codes are a natural generalization of cyclic codes,
they are ideals of the quotient algebra
$\F_q[X]/\langle X^n-\lambda\rangle$,
where $\lambda\in\F_q^*$ is a non-zero element.
Semisimple constacyclic codes are still the most important case.
In the semisimple case, \cite{DL} showed that
self-dual negacyclic codes (i.e., $\lambda=-1$) exist.
Blackford~\cite{B08} proved that semisimple self-dual constacyclic codes
do not exist except for the negacyclic case; and he showed
conditions that the self-dual negacyclic codes exist.
A more concise condition for the existence
of self-dual negacyclic codes appeared in \cite[Corollary 21]{CDFL}.

An innovation to develop further the theory is
the concept of {\em iso-dual constacyclic codes}
introduced in \cite{B13}.
Assume that $\lambda$ is a primitive $r$-th root of unity
(of course, $r$ divides $q-1$). Similarly to the cyclic codes,
the set of roots of $X^n-\lambda$ corresponds
to a subset of $\Z_{nr}$, denoted by $1+r\Z_{nr}$,
which consists of the elements $i\in\Z_{nr}$ such that $i\equiv 1\!\pmod r$.
The irreducible factors of $X^n-\lambda$ correspond bijectively
to the $q$-cosets on the set $1+r\Z_{nr}$. For $s\in\Z_{nr}^*$,
Blackford \cite{B13} defined a multiplier $\mu_s$ on $1+r\Z_{nr}$
similarly to Eqn~\eqref{multiplier},
and defined an isometry $\T_s$ of $\F_q[X]/\langle X^n-\lambda\rangle$
similarly to Eqn \eqref{isometry}. Then,
a constacyclic code $C$ is said to be {\em iso-dual}
if there is an isometry which maps $C$ onto its dual code $C^\bot$.
If there is a subset $P\subseteq(1+r\Z_{nr})$
which is a union of some $q$-cosets such that
$1+r\Z_{nr}= P\bigcup\mu_s(P)$ is a partition,
then $P$ and $\mu_s(P)$ are called a {\em {Type-I} duadic splitting}
of $1+r\Z_{nr}$ given by $\mu_s$;
in that case, $\F_q[X]/\langle X^n-1\rangle=C\oplus\T_s(C)$
is a direct sum, where $C$ is the constacyclic code with defining set~$P$.
Blackford \cite{B13} proved that iso-dual constacyclic codes
correspond one to one to the Type-I duadic splittings given by multipliers;
and showed conditions for the existence of such splittings
(an issue was addressed in \cite{CD}).
A necessary and sufficient condition for the existence of the
Type-I duadic splittings given by multipliers
was presented in \cite[Corollary 19]{CDFL}.
With these notations, however, the iso-dual cyclic codes do not exist either,
because any multiplier $\mu_s$ on $\Z_n$ leaves $0$ invariant,
hence the Type-I duadic splittings of $\Z_n$
given by multipliers do not exist;
see also Remark \ref{nu} (ii) below.

It is an interesting question: whether 
we can find a generalized self-duality
which occurs in semisimple cyclic codes?
The key idea is that, in essence, an isometry of cyclic codes
is an algebra isomorphism preserving Hamming weight;
though multipliers are isometries,
there might be isometries outside of the multipliers,
and those isometries might give isometrically self-dual cyclic codes.
We should investigate the general isometries of cyclic codes
to discover the isometries besides multipliers,
then we would answer the question positively.

In Section 2, we start with an investigation of
the general isometries of semisimple cyclic codes of length $n$,
and define the {\em isometrically self-dual cyclic codes},
abbreviated as {\em iso-self-dual cyclic codes}.
The permutations of $\Z_n$ corresponding the general isometries
include not only multipliers, and translations as well.
Then, we transform the study of iso-self-dual semisimple cyclic codes
into the study of Type-I duadic splittings of $\Z_n$
given by the multipliers and translations.

In Section 3, we present necessary and sufficient conditions
for the existence of Type-I duadic splittings of $\Z_n$
given by multipliers and translations. Meanwhile,
a method of construction of such splittings is also provided.

In Section 4, we exhibit a program to construct
iso-self-dual semisimple cyclic codes, and illustrate it with examples.
If the length $n$ is even and $n/2$ is odd, then for
any finite field $\F_q$ with $q$ coprime to $n$
the iso-self-dual cyclic codes of length~$n$ over $\F_q$ exist and the
construction of them is interesting.
In particular, a class of iso-self-dual MDS cyclic codes,
which are alternant codes from a class of generalized Reed-Solomon codes,
is obtained.

\section{Isometrically self-dual cyclic codes}

Let ${\F}_q$ be the finite field
with cardinality $|{\F}_q|=q$, where $q$ is a prime power,
let $\F_q^*$ be the set of non-zero elements of $\F_q$.
Let $n$ be a positive integer.

By $\Z_{ n}$ we denote the residue ring of the integer ring ${\Z}$
modulo $n$, and by $\Z_{n}^*$ we denote the multiplicative group
consisting of units of $\Z_{n}$.

By $R_n={\F}_q[X]/\langle X^n-1\rangle$ we denote
the quotient algebra of the polynomial algebra ${\F}_q[X]$
modulo the ideal $\langle X^n-1\rangle$ generated by $X^n-1$.
Any element of $R_{n}$ has a unique representative of degree  $<n$:
$a(X)=a_0+a_1X+\cdots+a_{n-1}X^{n-1}$, which is identified with
the word ${\bf a}=(a_0,a_1,\cdots,a_{n-1})\in {\F}_q^n$.
By ${\rm w}(a(X))$ we denote the Hamming weight of the word ${\bf a}$.
Any ideal $C$ of $R_n$ is called a {\em cyclic code}
of length~$n$ over ${\F}_q$.

In the following we {\em always} assume that $\gcd(q,n)=1$,
i.e., the cyclic codes considered in the following are all semisimple.
We always assume that $\theta$ is a primitive $n$-th root of unity
(in a suitable extension of $\F_q$). Then
\begin{equation}\label{in F_{q^d}}
 X^n-1
 =\prod\limits_{i\in {\Z}_{n}} (X-\theta^i).
\end{equation}
In this way, roots of $X^n-1$ correspond bijectively to
elements of $\Z_{n}$.

\begin{definition}\label{d:iso}
\rm
A transformation $\T:R_n\to R_n$ is called an {\em isometry}
of $R_n$ if $\T$ is an algebra automorphism of $R_n$ and
preserves the Hamming weight of $R_n$
(i.e., ${\rm w}(\T(a(X)))={\rm w}(a(X))$ for all $a(X)\in R_n$).
\end{definition}

General isometries between constacyclic codes were introduced
in \cite{CFLL} to classify constacyclic codes,
with less attention paid to cyclic codes.
We need precise characterizations of isometries of cyclic codes
as follows.

\begin{lemma}\label{l:iso}
A transformation
$\T:R_n\to R_n$ is an isometry if and only if
there is a positive integer $s\in{\Z}_n^*$ and an integer $t\in{\Z}_n$
satisfying that $qt\equiv t\!\pmod n$ such that
\begin{equation}\label{iso}
 \T\big(a(X)\big)=a(\theta^tX^s) 
 \!\!\pmod{X^n-1},~~~~~ \forall~a(X)=\sum_{i=0}^{n-1}a_iX^i\in R_n.
\end{equation}
\end{lemma}

\begin{proof}
Assume that $\T:R_n\to R_n$ is an isometry.
Since $\T$ preserves the Hamming weight,
there is an element $b\in{\F}_q^*$ and a positive integer 
$s\in{\Z}_n$ such that $\T(X)=b X^{s}$. 
Because $\T$ is an algebra automorphism, the $n$ elements:
\begin{equation}\label{basis}
 \T(X^i)=(bX^s)^i= b^iX^{si}\!\!\pmod{X^n-1}, ~~~~~
 i=0,1,\cdots,n-1,
\end{equation}
are a basis of $R_n$.
If the greatest common divisor $d=\gcd(s,n)>1$, then
the remainder of $si$ modulo $n$ is a multiple of $d$;
hence the $n$ elements in Eqn \eqref{basis} are not a basis of $R_n$,
which is not the case. So
the integer $s$ is coprime to $n$, i.e., $s\in\Z_n^*$.
Since $X^n=1$ in $R_n$, we have
$$
1=\T(X^n)=\T(X)^n=b^nX^{sn}=b^n;
$$
that is, $b=\theta^t$ for an integer $t\in\Z_n$.
By the Galois theory, $\theta^t\in\F_q$ if and only if
$(\theta^t)^q=\theta^t$ which is equivalent to that
$qt\equiv t\!\pmod n$. Then Eqn~\eqref{iso}
holds obviously.

Conversely, assume that $s\in{\Z}_n^*$, $t\in{\Z}_n$
satisfying that $qt\equiv t\!\pmod n$, and Eqn \eqref{iso} holds.
Note that $\theta^t\in\F_q$ because $qt\equiv t\!\pmod n$.
For any $a(X),b(X)\in R_n$, in $R_n$ we have the following computation:
$$
\T\big(a(X)b(X)\big)=
a(\theta^{t}X^{s})b(\theta^{t}X^{s})
=\T\big(a(X)\big)\T\big(b(X)\big)\!\pmod{X^n-1}.
$$
Similarly, we can check that
$\T\big(a(X)+b(X)\big)=\T\big(a(X)\big)+\T\big(b(X)\big)$.
Thus, $\T$ is an algebra endomorphism of $R_n$.
For $i=0,1,\cdots,n-1$, there are integers $h_i$
and $r_i$ with $0\le r_i\le n-1$ such that $si=nh_i+r_i$.
Then
\begin{equation*}
 \T(X^i)=\theta^{ti}X^{r_i}, ~~~~~ i=0,1,\cdots,n-1.
\end{equation*}
Since $s$ is coprime to $n$, it is easy to see
that $i\mapsto r_i$ is a permutation of $\{0,1,\cdots,n-1\}$.
Let $P_s$ be the corresponding permutation matrix.
By $D_t$ 
we denote the diagonal matrix with $1,\theta^{t},\cdots,\theta^{t(n-1)}$
on the diagonal positions.
Then $\T(X^i)=\theta^{ti}X^{r_i}$ implies that:
\begin{itemize}
\item {\it With respect to the basis $1,X,\cdots,X^{n-1}$ of $R_n$,
 the matrix of $\T$ is the monomial matrix $D_tP_s$.}
\end{itemize}
Hence $\T$ is an algebra automorphism
and preserves the Hamming weight.
\end{proof}

\begin{remark}\label{r:iso}\rm
For $s\in\Z_n^*$ and $t\in\Z_n$ with $qt\equiv t\!\pmod n$,
we have a positive integer
$s^{-1}\in\Z_n^*$ such that $ss^{-1}\equiv 1\!\pmod n$;
and for $-t$ we still have $q(-t)\equiv -t\!\pmod n$.
By $\T_{s,t}$ we denote the isometry of $R_n$ defined by
\begin{equation}\label{e:iso}
\T_{s,t}:~~ R_n \longrightarrow\, R_n,~~
 a(X) \longrightarrow\, a(\theta^{-t}X^{s^{-1}})\!\!\pmod{X^n-1}.
\end{equation}
That is, $\T_{s,t}\big(a(X)\big)=a(\theta^{-t}X^{s^{-1}})\!\pmod{X^n-1}$,
it is the unique polynomial of degree $<n$ congruent to
$a(\theta^{-t}X^{s^{-1}})$ modulo $X^n-1$.
We illustrate two specific cases.

(i)~ If $s=1$. then we take $s^{-1}=1$ and,
for any $a(X)=\sum_{i=0}^{n-1}a_iX^i\in R_n$, since
$\deg a(\theta^{-t}X)=\deg a(X)<n$, we have
\begin{equation}\label{e:trans iso}
\T_{1,t}\big(a(X)\big)=a(\theta^{-t}X)
=a_0+\theta^{-t} a_1X+\cdots+\theta^{-t(n-1)}a_{n-1}X^{n-1}.
\end{equation}

(ii)~ If $s\!=\!-1$ and $t\!=\!0$, then we take $s^{-1}\!=\!n-1$ as
$(-1)(n-1)\!\equiv\! 1\!\pmod n$;
since $X^n\equiv 1\!\pmod{X^n-1}$,
by Eqn \eqref{e:iso}, for $a(X)=\sum_{i=0}^{n-1}a_iX^i$ we have
$$\T_{-1,0}\big(a(X)\big)
 \equiv X^n\sum_{i=0}^{n-1}a_iX^{(n-1)i} 
\equiv\sum_{i=0}^{n-1}a_iX^{n-i}\pmod{X^n-1} ;$$
i.e.,
\begin{equation}\label{e:-1 iso}
\T_{-1,0}\big(a(X)\big)=a_0+a_{n-1}X+\cdots+a_1X^{n-1}.
\end{equation}
\end{remark}\medskip

As usual, by $C^\bot$ we denote the
(Euclidean) dual code of any code $C$.
It is easy to check that, for any cyclic code $C\subseteq R_n$,
its dual code $C^\bot\subseteq R_n$ is still a cyclic  code.

\begin{definition}\rm
For a cyclic code $C\subseteq R_n$ over $\F_q$ of length $n$,
if there is an isometry $\T_{s,t}$ of $R_n$ such that
$\T_{s,t}(C)=C^\bot$, then we say that
$C$ is an {\em isometrically self-dual cyclic code},
or, an {\em iso-self-dual cyclic code} for short.
\end{definition}

Thus, an iso-self-dual cyclic code $C$ has the same
algebraic structure as its dual code $C^\bot$; specifically,
$h(X)$  is a check polynomial of $C$ if and only if
 $\T_{s,t}\big(h(X)\big)$ is a check polynomial of $\T_{s,t}(C)=C^\bot$.
Meanwhile, an iso-self-dual cyclic code $C$
has the same weight structure as $C^\bot$;
specifically, $C$ and $C^\bot$ have the same
weight distributions; in other words, $C$ is a {\em formally self-dual}
cyclic code, see \cite[page 378]{HP}.

\medskip
In the rest of this section we translate the study of iso-self-dual
cyclic codes into a study of corresponding permutations
on the index set $\Z_n$.

For any pair $s,t$ where $s\in\Z_n^*$ and $t\in\Z_n$ satisfying
that $qt\equiv t\!\pmod n$, it is easy to see that the map
\begin{equation}\label{q-permutation}
\P_{s,t}:~~ \Z_n \longrightarrow \Z_n,~~
 i \longrightarrow s(i+t)\!\pmod n,
\end{equation}
is a permutation of $\Z_n$, which we call a {\em $q$-permutation of $\Z_n$}.

We list some related known facts in the following remark.

\begin{remark}\label{r:known cyclic}\rm
(i)~
Note that, if $t=0$, then we denote $\mu_s=\P_{s,0}$,
which is called a {\em multiplier} on $\Z_n$ in literature.
On the other hand, if $s=1$, then
we denote $\tau_t=\P_{1,t}$.
Since $\tau_t(i)=i+t\!\pmod n$ for $i\in\Z_n$
and $t$ should satisfy that ``$qt\equiv t\!\pmod n$'' where $q$ is involved,
we call $\tau_t$ a {\em $q$-translation} on $\Z_n$.
Obviously, $\P_{s,t}=\mu_s\tau_t$.

(ii)~
Since $q\in\Z_n^*$, we have a special multiplier $\mu_q$,
and the $\mu_q$-orbits
(i.e., the orbits on $\Z_n$ of the permutation group $\langle\mu_q\rangle$
generated by $\mu_q$)
are usually named {\em $q$-cyclotomic cosets}.
In the following, we call $\mu_q$-orbits on $\Z_n$ by {\em $q$-cosets},
and use $\Z_n/\mu_q$ to stand for the quotient sets
consisting of all $q$-cosets on $\Z_n$.

(iii)~ For any $q$-coset $Q\in\Z_n/\mu_q$,
the polynomial $f_Q(X)=\prod_{i\in Q}(X-\theta^i)$ is
irreducible in $\F_q[X]$.
We have the monic irreducible decomposition:
$X^n-1=\prod_{Q\in\Z_n/\mu_q}f_Q(X)$.
Obviously, a subset $P\subseteq \Z_n$ is $\mu_q$-invariant
(i.e., $\mu_q(P)=P$) if and only if
$P=\bigcup_{j=1}^k Q_j$ is a union of some $q$-cosets
$Q_1,\cdots,Q_k$. If it is the case, then the polynomial
$f_P(X)=\prod_{i\in P}(X-\theta^i)=\prod_{j=1}^k f_Q(X)$
is a monic divisor of $X^n-1$. In this way,
the monic divisors of $X^n-1$ correspond one to one to the
$\mu_q$-invariant subsets of $\Z_n$.

(iv)~ Each $\mu_q$-invariant subset $P$ of $\Z_n$
corresponds to exactly one cyclic code (i.e., ideal) of $R_n$, denoted by $C_P$,
such that the polynomial $f_P(X)$ is a check polynomial of $C_P$;
both the degree $\deg f_P(X)$ and the dimension $\dim C_P$
equal to the cardinality $|P|$ of $P$.
The correspondence $P\mapsto C_P$ is a lattice isomorphism
between the lattice of $\mu_q$-invariant subsets of $\Z_n$
and the lattice of cyclic codes of $R_n$.
\end{remark}

\begin{lemma}\label{l:T Q}
Let $s\in\Z^*_n$ and $t\in\Z_n$ satisfying that $qt\equiv t\!\pmod n$,
and $s^{-1}$ be a positive integer such that $ss^{-1}\equiv 1\!\pmod{n}$.
Then for $f_Q(X)$, where $Q\in\Z_n/\mu_q$ is a $q$-coset,
we have an $\F_q$-polynomial $u_Q(X)$ which is coprime to $X^n-1$
and has leading coefficient $\theta^{-t|Q|}$ such that
$$
  f_Q(\theta^{-t}X^{s^{-1}}) =u_Q(X) f_{\P_{s,t}(Q)}(X) .
$$
\end{lemma}

\begin{proof}
Since $qt\equiv t\!\pmod n$, we see that $\theta^{-t}\in\F_q$ and
$$\P_{s,t}(Q)=s(Q+t)=\{s(i+t)\!\pmod n \mid i\in Q\}$$
is still a $q$-coset, and the cardinalities $\big|Q\big|=\big|\P_{s,t}(Q)\big|$.
Then:
\begin{eqnarray*}
 f_Q(\theta^{-t}X^{s^{-1}})
&=& \prod_{i\in Q}(\theta^{-t}X^{s^{-1}}-\theta^i)
   =\prod_{i\in Q}\theta^{-t}\big(X^{s^{-1}}-(\theta^{(i+t)s})^{s^{-1}}\big)\\
&=&\theta^{-t|Q|}\prod_{i\in Q}
        \frac{X^{s^{-1}}-(\theta^{(i+t)s})^{s^{-1}}}{X-\theta^{(i+t)s}}
        \cdot\prod_{i\in Q}(X-\theta^{(i+t)s}).
\end{eqnarray*}
Since the map $Q\to \P_{s,t}(Q)$, $i\mapsto s(i+t)\!\pmod n$, is a bijection, we see that
$$
\prod_{i\in Q}(X-\theta^{(i+t)s})
=\prod_{i\in\P_{s,t}(Q)}(X-\theta^{i})
=f_{\P_{s,t}(Q)}(X) .
$$
It remains to show that  $\prod_{i\in Q}
  \frac{X^{s^{-1}}-(\theta^{(i+t)s})^{s^{-1}}}{X-\theta^{(i+t)s}}$
is coprime to $X^n-1$.
To do it, it is enough to show that any root $\theta^j$ of $X^n-1$
is not a root of the polynomial
$\frac{X^{s^{-1}}-(\theta^{(i+t)s})^{s^{-1}}}{X-\theta^{(i+t)s}}$
for any $i\in Q$.
If $j{\equiv}(i+t)s\!\pmod n$,
then $\theta^j=\theta^{(i+t)s}$ which is not a root of the polynomial
$\frac{X^{s^{-1}}-(\theta^{(i+t)s})^{s^{-1}}}{X-\theta^{(i+t)s}}$.
Otherwise, $j{\not\equiv}(i+t)s\!\pmod n$
hence $\theta^j\ne \theta^{(i+t)s}$; since $s^{-1}$ is coprime to $n$,
$(\theta^j)^{s^{-1}}-(\theta^{(i+t)s})^{s^{-1}}\ne 0$; i.e.,
$\theta^j$ is not a root of the polynomial
$\frac{X^{s^{-1}}-(\theta^{(i+t)s})^{s^{-1}}}{X-\theta^{(i+t)s}}$.
\end{proof}

\begin{corollary}\label{c:T P}
Let $s,t, s^{-1}$ be as in Lemma \ref{l:T Q}, and
$P\subseteq\Z_n$ be a $\mu_q$-invariant subset. Then
we have an $\F_q$-polynomial $u_P(X)$ which is coprime to $X^n-1$
and has leading coefficient $\theta^{-t|P|}$ such that
$$
 f_P(\theta^{-t}X^{s^{-1}}) =u_P(X) f_{\P_{s,t}(P)}(X) .
$$
\end{corollary}

\begin{proof} There are some $q$-cosets $Q_1,\cdots,Q_k$ such that
$P$ is a disjoint union $P=\bigcup_{i=1}^k Q_i$. Then
$\P_{s,t}(P)$ is a disjoint union $\P_{s,t}(P)=\bigcup_{i=1}^k \P_{s,t}(Q_i)$,
and $f_P(X)=\prod_{i=1}^k f_{Q_i}(X)$.
By Lemma \ref{l:T Q},
$f_{Q_i}(\theta^{-t}X^{s^{-1}})=u_{Q_i}(X)f_{\P_{s,t}(Q_i)}(X)$
with $u_{Q_i}(X)$ being coprime to $X^n-1$
and having leading coefficient $\theta^{-t|Q_i|}$. So
$$  f_P(\theta^{-t}X^{s^{-1}}) 
  =\prod_{i=1}^k u_{Q_i}(X) \cdot \prod_{i=1}^k f_{\P_{s,t}(Q_i)}(X)
  =\prod_{i=1}^k u_{Q_i}(X) \cdot  f_{\P_{s,t}(P)}(X),
$$
where $\prod_{i=1}^k u_{Q_i}(X)$ is coprime to $X^n-1$
and has leading coefficient $\theta^{-t|P|}$.
\end{proof}

Though the polynomials $\T_{s,t}\big(f_P(X)\big)$ and
$f_{\P_{s,t}(P)}(X)$ are different from each other in general,
by the above result, we can find the polynomial $f_{\P_{s,t}(P)}(X)$
from the polynomial $\T_{s,t}\big(f_P(X)\big)$.

\begin{theorem}\label{t:T P}
Let $s\in\Z_n^*$ and $t\in\Z_n$ satisfying that $qt\equiv t\!\pmod{n}$,
let $P$ be a $\mu_q$-invariant subset of $\Z_n$. Then
$$
 f_{\P_{s,t}(P)}(X)=\gcd\big( f_P(\theta^{-t}X^{s^{-1}}),\,X^n-1\big)
 =\gcd\big(\T_{s.t}\big( f_P(X)\big),\,X^n-1\big).
$$
\end{theorem}

\begin{proof}
From Corollary \ref{c:T P}, we get immediately
that
$$
f_{\P_{s,t}(P)}(X)=\gcd\big( f_P(\theta^{-t}X^{s^{-1}}),\,X^n-1\big).
$$
Note that $\T_{s,t}\big(f_P(X)\big)$ is the remainder of
$f_P(\theta^{-t}X^{s^{-1}})$ modulo $X^n-1$.
The greatest common divisor in the right hand side of the above equality
is clearly equal to $\gcd\big(\T_{s.t}\big( f_P(X)\big),\,X^n-1\big)$.
\end{proof}

We describe two specific cases of the theorem.
Recall that, for a polynomial
$a(X)=a_0+a_1X+\cdots+a_dX^d$,
the polynomial $X^d\cdot a(X^{-1})=a_d+a_{d-1}X+\cdots+a_0X^d$ is
called the {\em reciprocal polynomial} of $a(X)$.

\begin{corollary}\label{c:T1 P}
(i)~
$f_{\P_{1,t}(P)}(X)=\theta^{t|P|}f_P(\theta^{-t}X)$.

(ii)~ Let $a_{0}=(-1)^{|P|}\prod_{i\in P}\theta^{i}$
be the constant term of $f_P(X)$. Then
$$f_{\P_{-1,0}(P)}(X)=f_{-P}(X)
 =a_{0}^{-1}\cdot X^{|P|}f_P(X^{-1}).
$$
\end{corollary}

\begin{proof}
Note that both $f_P(X)$ and $f_{\P_{s,t}(P)}(X)$ are
monic divisors of $X^n-1$ with degree $|P|$.

(i).~ In this case where $s=1$ (hence $s^{-1}=1$),
the polynomial $u_P(X)$ in Corollary \ref{c:T P}
is in fact the constant $\theta^{-t|P|}$.

(ii).~ Note that $\P_{-1,0}(P)=(-1)(P+0)=-P=\{-i\mid i\in P\}$.
Denote $|P|=d$. We can write
$f_P(X)=a_0+a_1X+\cdots+a_{d}X^{d}$,
where $a_d=1$ and $a_0=(-1)^d\prod_{i\in P}\theta^i$
because $f_P(X)=\prod_{i\in P}(X-\theta^i)$.
By Eqn \eqref{e:-1 iso},
\begin{eqnarray*}
 \T_{-1,0}\big(f_P(X)\big)
 &=&a_0+a_{1}X^{n-1}+\cdots+a_dX^{n-d}\\
 &=&X^{n-d}(a_0X^d+a_{1}X^{d-1}+\cdots+a_d) -a_0(X^n-1)\\
 &=& X^{n-d}\cdot X^df_P(X^{-1}) -a_0(X^n-1).
\end{eqnarray*}
Since $\theta^i$ for $i\in(-P)$ are all roots of $X^d f_P(X^{-1})$
(cf. \cite[Lemma 4.4.7]{HP}),
$X^d f_P(X^{-1})$ is a divisor of $X^n-1$. Thus
$$
\gcd\big(\T_{-1.0}\big( f_P(X)\big),\,X^n-1\big)
=a_0^{-1}(a_0X^d+a_{1}X^{d-1}+\cdots+a_d).
$$
By Theorem \ref{t:T P}, we obtain the conclusion (ii).
\end{proof}

We turn to cyclic codes and state the following concise result.

\begin{theorem}\label{t:T and P}
For any $q$-invariant subset $P\subseteq\Z_n$ and any isometry $\T_{s,t}$
$$
 \T_{s,t}( C_P) = C_{\P_{s,t}(P)} .
$$
\end{theorem}

\begin{proof}
Since $f_P(X)$ is a check polynomial of $C_P$ and $\T_{s,t}$
is an algebra automorphism of $R_n$,
the cyclic code $\T_{s,t}(C_P)$ has a check polynomial
$\T_{s,t}\big(f_P(X)\big)=f_P(\theta^{-t}X^{s^{-1}})\!\pmod{X^n-1}$.
By Corollary \ref{c:T P}, $u_P(X) f_{\P_{s,t}(P)}(X)$
is a check polynomial of $\T_{s,t}(C_P)$,
where $u_P(X)$ is a unit of $R_n$.
So, $f_{\P_{s,t}(P)}(X)$ is also a check polynomial of $\T_{s,t}(C_P)$.
By definition (see Remark \ref{r:known cyclic} (iv)),
$\T_{s,t}( C_P) = C_{\P_{s,t}(P)}$.
\end{proof}

Similar to \cite[Theorem 4]{B13} and \cite[Lemma 3.5]{FZ16},
we have the following result.

\begin{lemma}\label{l:dual code}
Let $P$ be a $q$-invariant subset of $\Z_n$. Then
$$C_P^{\bot}=\T_{-1,0}(C_{\bar P})=C_{-\bar P}\,,$$
where $\bar P\!=\!\Z_n\backslash P$ is the complement of $P$
in $\Z_n$.
\end{lemma}

\begin{proof} By Theorem \ref{t:T and P},
$\T_{-1,0}(C_{\bar P})=C_{\P_{-1,0}(\bar P)}=C_{-\bar P}$.

Let $a(X)=\sum_{i=0}^{n-1}a_iX^i\in C_{\bar P}$.
For any $c(X)=\sum_{i=0}^{n-1}c_iX^i\in C_P$,
set $b(X)=X^{n-1}c(X)\!\pmod{X^n-1}$. Then $b(X)\in C_P$
(as $C_P$ is an ideal) and
$$b(X)=c_1+c_2X+\cdots+c_{n-1}X^{n-2}+c_0X^{n-1}.$$
Since $f_P(X)f_{\bar P}(X)=X^n-1$,
$a(X)b(X)\equiv 0\!\pmod{X^n-1}$. Observing the coefficient of $X^{n-1}$
in $a(X)b(X)\!\pmod{X^n-1}$, we obtain that
$$
a_0c_0+a_1c_{n-1}+\cdots+a_{n-2}c_2+a_{n-1}c_1=0.
$$
By Eqn \eqref{e:-1 iso},
$$\T_{-1,0}\big(a(X)\big)=a_0+a_{n-1}X+\cdots+a_1X^{n-1}.$$
So the (Euclidean) inner product
$$
\big\langle\T_{-1,0}\big(a(X)\big),c(X)\big\rangle
=a_0c_0+a_{n-1}c_1+\cdots+a_{1}c_{n-1}=0
$$
Thus $\T_{-1,0}\big(a(X)\big)\in C_P^\bot$.
In conclusion, $\T_{-1,0}(C_{\bar P})\subseteq C_P^\bot$.
Finally,
$$\dim \T_{-1,0}(C_{\bar P})=\dim C_{\bar P}
   =|\bar P|=n-|P|=n-\dim C_P=\dim C_P^\bot;$$
from which we further get that $\T_{-1,0}(C_{\bar P})=C_P^\bot$.
\end{proof}

\begin{theorem}\label{t:dual splitting}
Let $P$ be a $q$-invariant subset of $\Z_n$. The cyclic code
$C_P\subseteq R_n$ is an iso-self-dual cyclic code
if and only if there is a positive integer $s\in\Z_n^*$ and an integer
$t\in\Z_n$ satisfying that $qt\equiv t\!\pmod n$ such that
$\Z_n=P\bigcup \P_{s,t}(P)$ is a partition.
If it is the case, then the following hold:
\begin{itemize}
\item[\rm(i)]
$\T_{-s,t}(C_P)=C_P^\bot$.
\item[\rm(ii)]
$h(X)$ is a check polynomial of $C_P$ if and only if
$\T_{-s,t}(h(X))$ is a check polynomial of $C_P^\bot$
(equivalently, $\T_{-s,t}(h(X))$ is a generator polynomial of $C_P$).
\item[\rm(iii)]
$C_P$ and $C_P^\bot$ have the same weight distributions.
\end{itemize}
\end{theorem}

\begin{proof} Let $\T_{s,t}$ be an isometry of $R_n$.
By Theorem \ref{t:T and P} and Lemma \ref{l:dual code},
 $\T_{s,t}(C_P)=C_P^\bot$ is equivalent to
 $C_{\P_{s,t}(P)}=C_{-\bar P}$, i.e.,
$-\bar P=\P_{s,t}(P)=s(P+t)$.
Thus  $\T_{s,t}(C_P)=C_P^\bot$ if and only if
$\bar P=-s(P+t)=\P_{-s,t}(P)$.
The latter one is equivalent to that
$\Z_n=P\bigcup \P_{-s,t}(P)$ is a partition of $\Z_n$.

Replacing $-s$ by $s$ in the above argument, we obtain the conclusions (i).
Since $\T_{-s,t}$ is an algebra automorphism of $R_n$
and preserves Hamming weight,
(ii) and (iii) are obtained immediately.
\end{proof}

\begin{definition}\label{splitting}\rm
(i)~
Let $s\in\Z_n^*$, $t\in\Z_n$ satisfying that $qt\equiv t\!\pmod n$,
and $P$ be a  $\mu_q$-invariant subset of $\Z_n$.
If $\Z=P\bigcup\P_{s,t}(P)$ is a partition of $\Z_n$,
then we say that $P,\P_{s,t}(P)$ are a {\em Type-I duadic splitting of $\Z_n$}
given by $\P_{s,t}$.

(ii)~ If there are $s,t,P$ such that the above (i) holds, then we say that
Type-I duadic splittings of $\Z_n$ exist.
\end{definition}

Thus, the study of iso-self-dual cyclic codes is translated
to the study of Type-I duadic splittings of $\Z_n$.
Theorem~\ref{t:dual splitting} says that
iso-self-dual cyclic codes of length $n$ exist if and only if
Type-I duadic splittings of $\Z_n$ exist; and, in that case,
the information of the iso-self-dual cyclic codes
can be obtained from the information of the
Type-I duadic splittings of~$\Z_n$.

\section{Type-I duadic splittings of $\Z_n$}

Keep the notations in Section 2. For convenience, we denote
\begin{equation}\label{G_q,n}
 G_{q,n}=\big\{\P_{s,t}\,\big|\,
    s\in\Z_n^*,~ t\in\Z_n\mbox{ satisfying that } qt\equiv t\!\!\!\pmod n\big\},
\end{equation}
where $\P_{s,t}$ is the $q$-permutation
of $\Z_n$ defined in Eqn \eqref{q-permutation}.
It is easy to check that
\begin{equation*}
\P_{s',t'}\P_{s,t}=\P_{s's,t+s^{-1}t'},~~~~~ \forall~\P_{s',t'},\P_{s,t}\in G_{q,n};
\end{equation*}
in particular,
\begin{equation*}
 \P_{s,t}^{-1}=\P_{s^{-1},-st},~~~~~ \forall~\P_{s,t}\in G_{q,n}.
\end{equation*}
So, $G_{q,n}$ is a permutation group of $\Z_n$,
which we call the {\em $q$-permutation group} of $\Z_n$.
Recall from Remark \ref{r:known cyclic} (i) that
we denote $\mu_s=\P_{s,0}$ and call it a {\em multiplier};
we denote $\tau_t=\P_{1,t}$ and call it a {\em $q$-translation} on $\Z_n$;
and $\P_{s,t}=\mu_s\tau_t$.

From the condition $qt\equiv t\!\pmod n$, we
have $\tau_t\mu_q=\mu_q\tau_t$.
So the subgroup $\langle\mu_q\rangle$ generated by $\mu_q$
is a central subgroup of $G_{q,n}$.
Then the group $G_{q,n}$ acts on the quotient set
$\Z_n/\mu_q$ consisting of $q$-cosets.
The properties of this action determine
the existence and constructions of the Type-I duadic splittings of $\Z_n$.
For fundamentals of group actions, please refer to \cite{AB}.
Some elementary results on group actions will be cited later
which we list in the following remark; cf. \cite{AB} or
\cite[Lemma 6, Lemma 7]{CDFL} for details.

\begin{remark}\label{r:G action}\rm
Let a finite group $G$ act on a finite set $X$.
For $\P\in G$, by $\langle\P\rangle$ we denote the subgroup generated by $\P$.
Any orbit of the group $\langle \P\rangle$ on $X$ is abbreviated as a $\P$-orbit.

(i)~ Let $\P\in G$. Then there is a subset $S\subseteq X$ such that
$X=S\bigcup\P(S)$ is a partition of $X$ if and only if
the length of any $\P$-orbit on $X$ is even.
If it is the case, for each $\P$-orbit $(x_1,x_2,\cdots,x_{2m})$,
i.e., $\P(x_i)=x_{i+1}$ for $1\le i<2m$ and $\P(x_{2m})=x_{1}$,
we can place  $x_1,x_3,\cdots,x_{2m-1}$ into $S$;
once this is done for every $\P$-orbit, we obtain
a subset $S\subseteq X$ such that $X=S\bigcup\P(S)$ is a partition.

(ii)~ The length of any $\P$-orbit on $X$ is a divisor of
the order of $\rho$ in the group $G$. In particular,
if the order of $\P$ is a power of $2$, then
there is a subset $S\subseteq X$ such that
$X=S\bigcup\P(S)$ is a partition of $X$ if and only if
there is no element of $X$ which is fixed by $\P$.

(iii)~ Let another finite group $G'$ act on a finite set $X'$.
Then the finite group $G\times G'$ acts on the finite set $X\times X'$.
For $(\P,\P')\in G\times G'$ and $(x.x')\in X\times X'$,
the length of the $(\P,\P')$-orbit on $X\times X'$ containing $(x,x')$
equals to the least common multiple of
the lengths of the $\P$-orbit on $X$ containing $x$ and
the length of the $\P'$-orbit on $X'$ containing $x'$.
\end{remark}

\begin{lemma}\label{l:Q fixed}
Let $\P_{s,t}\in G_{q,n}$.
A $q$-coset $Q\in\Z_n/\mu_q$ is fixed by $\P_{s,t}$ if and only if
there is a $k\in Q$ and an integer $j\ge 0$ such that
$(q^j-s)k\equiv st\!\pmod n$.
\end{lemma}

\begin{proof} Note that, since $\P_{s,t}(Q)$ is still a $q$-coset,
$\P_{s,t}(Q)=Q$ if and only if $\P_{s,t}(Q)\bigcap Q\ne\emptyset$.
Thus, $\P_{s,t}(Q)=Q$ if and only if there is a $k\in Q$ and an integer
$j\ge 0$ such that $\P_{s,t}(k)=q^jk$.
The latter is just $(q^j-s)k\equiv st\!\pmod n$.
\end{proof}

\begin{remark}\label{nu} \rm
(i)~ By $\nu_2(m)$ we denote the $2$-adic valuation
of any non-zero integer $m$,
i.e., $2^{\nu_2(m)}$ is the largest power of $2$ which divides $m$.
For convenience, we appoint that $\nu_2(0)=\infty$.

(ii)~ If $t=0$, then
``$k=0$'' is always a solution for $(q^j-s)k\equiv st\!\pmod n$;
hence the Type-I duadic splittings of $\Z_n$ given
by any multiplier $\P_{s,0}=\mu_s$ do not exist.
Thus, in the following we consider the case
where $t\,{\not\equiv}\,0\!\pmod n$.

(iii)~
From Lemma \ref{l:t,s,duadic} till to Theorem \ref{t:duadic},
we consider the case where $n=2^v$ with $v\ge 1$.
In that case, $q$ must be odd since we have assumed that
${\gcd(q,n)=1}$, and,
with the notation in (i),
the condition that $qt\equiv t\!\pmod n$ in Eqn~\eqref{G_q,n}
(i.e. $(q-1)t\equiv 0\!\pmod{2^v}$) can be rewritten as
$v\le\nu_2(q-1)+\nu_2(t)$.
\end{remark}

\begin{lemma}\label{l:t,s,duadic}
Let $n=2^v$ with $v\ge 1$.
Let $\P_{s,t}\in G_{q,2^v}$ with $t\,{\not\equiv}\,0\!\pmod{2^v}$.
The Type-I duadic splittings of $\Z_{2^v}$ given by $\P_{s,t}$
exist if and only if
\begin{equation}\label{t,s,duadic}
\nu_2(q^j-s)>\nu_2(t),~~~~~~~~ \forall~j\ge 0.
\end{equation}
\end{lemma}

\begin{proof}
Since two integers in a residue class of $\Z_{2^v}$ have different
$2$-adic valuations, we need to show that the 
inequality \eqref{t,s,duadic} is independent of the choices 
of the integers $s,t$ from their residue classes modulo $2^v$.
Suppose that $s\in\Z_{2^v}^*$ satisfies the inequality \eqref{t,s,duadic},
i.e., $\nu_2(q^j-s)>\nu_2(t)$.
For any $s'\equiv s\!\pmod {2^v}$, we write $s'=s+2^vk$.
Since $0\le \nu_2(t)<v$, we get
$$
\nu_2(q^j-s')=\nu_2(q^j-s-2^vk)
\ge \min\{\nu_2(q^j-s),v\}>\nu_2(t).
$$
So, $s'$ still satisfies the inequality \eqref{t,s,duadic}.
Similarly, if $\nu_2(q^j-s)>\nu_2(t)$ and $t'=t+2^v k$,
then $\nu_2(t')=\nu_2(t)$ since $\nu_2(t)<v$, hence
$\nu_2(q^j-s)>\nu_2(t')$.

It is easy to see that the group $G_{q,2^v}$ is of order power of $2$,
hence the order of $\P_{s,t}$ is a power of $2$.
Thus the Type-I duadic splittings of $\Z_{2^v}$ given by $\P_{s,t}$
exist if and only if there is no $q$-coset of $\Z_{2^v}$ which is fixed by $\P_{s,t}$,
please see Remark \ref{r:G action} (ii).

Suppose that the condition \eqref{t,s,duadic} holds.
Then for any $j\ge 0$,
$$\nu_2((q^j-s)k)\ge \nu_2(q^j-s)>\nu_2(t)=\nu_2(st),
  ~~~~~~ \forall~k\in\Z_n.$$
Since $t\,{\not\equiv}\,0\!\pmod{2^v}$,
for any $k\in\Z_n$ we get that
$(q^j-s)k\,{\not\equiv}\,st\!\pmod{2^v}$.
By Lemma~\ref{l:Q fixed},
there is no $q$-coset $Q\in\Z_n/\mu_q$ which is fixed by $\P_{s,t}$.

Conversely, assume that there is a $j\ge 0$ such that
$\nu_2(q^j-s)\le\nu_2(t)=\nu_2(st)$. Then we can find an integer $k\in\Z_{2^v}$
such that $(q^j-s)k\equiv st\!\pmod{2^v}$.
By Lemma \ref{l:Q fixed}, the $q$-coset $Q_k$
containing $k$ is fixed by $\P_{s,t}$.
\end{proof}

\begin{corollary}\label{c:t,duadic}
Let $n=2^v$ with $v\ge 1$. Let $\P_{1,t}=\tau_t\in G_{q,2^v}$
with $\nu_2(t)<v$ be a $q$-translation.
The Type-I duadic splittings of $\Z_{2^v}$ given by $\tau_{t}$
exist if and only if $$\nu_2(t)<\nu_2(q-1).$$
\end{corollary}

\begin{proof}
If the Type-I duadic splittings of $\Z_{2^v}$ given by $\tau_{t}$ exist,
by Lemma~\ref{l:t,s,duadic},
$\nu_2(t)<\nu_2(q^j-1)$ for any $j\ge 0$. Taking $j=1$, we obtain
$\nu_2(t)<\nu_2(q-1)$.

Conversely, assume that $\nu_2(t)<\nu_2(q-1)$.
Since $(q-1)\mid(q^j-1)$ for any $j\ge 0$,
we get $\nu_2(t)<\nu_2(q^j-1)$ for any $j\ge 0$.
By Lemma \ref{l:t,s,duadic},
the Type-I duadic splittings of $\Z_{2^v}$ given by $\tau_{t}$ exist.
\end{proof}

\begin{lemma}\label{l:s,duadic}
Let $n=2^v$ with $v\ge 1$, and $s\in\Z_{2^v}^*$.
There is a $t\in\Z_{2^v}$ with $qt\equiv t\!\pmod {2^v}$
such that Type-I duadic splittings of $\Z_{2^v}$ given by $\P_{s,t}$
exist if and only if
\begin{equation}\label{s,duadic}
 \nu_2(q^j-s)+\nu_2(q-1)>v, ~~~~~~~~ \forall~j\ge 0.
\end{equation}
\end{lemma}

\begin{proof}
First we need to show that the inequality \eqref{s,duadic}
is independent of the choice of the integer $s$ from its residue class modulo $2^v$.
Suppose that $s\in\Z_{2^v}^*$ satisfies the inequality \eqref{s,duadic},
i.e., $\nu_2(q^j-s)>v-\nu_2(q-1)$.
For any $s'\equiv s\!\pmod {2^v}$ we can write $s'=s+2^vk$.
Since $\nu_2(q-1)>0$ (recall that $q$ is odd), we get
$$
\nu_2(q^j-s')=\nu_2(q^j-s-2^vk)
\ge \min\{\nu_2(q^j-s),v\}>v-\nu_2(q-1).
$$
That is, $s'$ still satisfies the inequality \eqref{s,duadic}.

If $\nu_2(q-1)\ge v$, by the condition ``$\nu_2(q-1)+\nu_2(t)\ge v$''
(see Remark~\ref{nu}(iii)),
we can take a $t\in\Z_{2^v}$ such that $\nu_2(t)=0$, i.e., $t$ is odd.
Since both $q$ and $s$ are odd, for any $j\ge 0$,
$\nu_2(q^j-s)\ge 1 >\nu_2(t)$; hence
the condition \eqref{s,duadic} is satisfied, and
by Lemma \ref{l:t,s,duadic}, Type-I duadic splittings of $\Z_{2^v}$
given by $\P_{s,t}$ exist. In a word, the conclusion of the lemma holds in
this case.

Next, we assume that $\nu_2(q-1)<v$, i.e.,
$0<v-\nu_2(q-1)<v$.

Suppose that the condition \eqref{s,duadic} is satisfied.
Since $\nu_2(q-1)<v$, we can take a $t\in\Z_{2^v}$
such that $\nu_2(t)=v-\nu_2(q-1)$, see Remark~\ref{nu}(iii). Then
$$
\nu_2(q^j-s)>v-\nu_2(q-1)=\nu_2(t), ~~~~~~~~ \forall~ j\ge 0;
$$
by Lemma \ref{l:t,s,duadic}, the Type-I duadic splittings of $\Z_{2^v}$
given by $\P_{s,t}$ exist.

Suppose that the condition \eqref{s,duadic} is not satisfied,
i.e., there is a $j\ge 0$ such that $\nu_2(q^j-s)\le v-\nu_2(q-1)$.
For any $\P_{s,t}\in G_{q,n}$, by Remark~\ref{nu}(iii),
$\nu_2(q-1)+\nu_2(t)\ge v$.
So, $v-\nu_2(q-1)\le\nu_2(t)<v$. Then $\nu_2(q^j-s)\le v-\nu_2(q-1)\le\nu_2(t)$.
By Lemma \ref{l:t,s,duadic},
Type-I duadic splittings of $\Z_{2^v}$ given by $\P_{s,t}$
do not exist.
\end{proof}

\begin{lemma}\label{l:t,duadic}
Let $n=2^v$ with $v\ge 1$, let $\P_{s,t}\in G_{q,n}$.
If the Type-I duadic splittings of $\Z_{2^v}$ given by $\P_{s,t}$ exist,
then the Type-I duadic splittings of $\Z_{2^v}$ given by $\P_{1,t}=\tau_t$ exist.
\end{lemma}

\begin{proof}
By Lemma~\ref{l:t,s,duadic}, for $\P_{s,t}$
the condition \eqref{t,s,duadic} is satisfied.
If $s\equiv 1\!\pmod{2^v}$, then we are done.
In the following we assume that
$s\,{\not\equiv}\, 1\!\pmod{2^v}$.
For $j=0$, from the condition \eqref{t,s,duadic} we get
$$
 \nu_2(1-s)>\nu_2(t).
$$
Taking $j=1$ we further get $\nu_2(q-s)>\nu_2(t)$,
which we can rewrite as follows:
$$
 \nu_2\big((q-1)+(1-s)\big)>\nu_2(t).
$$
If $\nu_2(q-1)<\nu_2(1-s)$, then
$$
\nu_2(q-1)= \nu_2\big((q-1)+(1-s)\big)>\nu_2(t).
$$
Otherwise, $\nu_2(q-1)\ge \nu_2(1-s)$;
since $\nu_2(1-s)>\nu_2(t)$, we still get
$$
\nu_2(q-1)>\nu_2(t).
$$
For any $j\ge 0$, noting that $(q-1)\mid(q^j-1)$, we obtain that
$$
\nu_2(q^j-1)\ge \nu_2(q-1)>\nu_2(t).
$$
By Lemma \ref{l:t,s,duadic},
the Type-I duadic splittings of $\Z_{2^v}$ given by $\P_{1,t}$ exist.
\end{proof}

\begin{theorem}\label{t:duadic}
Let $n=2^v$. 
Then the following three statements
are equivalent to each other:
\begin{itemize}
\item[\rm(i)]
The Type-I duadic splittings of $\Z_{2^v}$ exist.
\item[\rm(ii)]
 The Type-I duadic splittings of $\Z_{2^v}$ given by $q$-translations exist.
\item[\rm(iii)]
 $0<v<2\cdot \nu_2(q-1)$.
\end{itemize}
If it is the case, then there is an integer $u$ such that
\begin{equation}\label{e:1,t}
\max\{0,~v-\nu_2(q-1)\}\le u <\min\{v,~\nu_2(q-1)\};
\end{equation}
hence $\P_{1,2^u}=\tau_{2^u}\in G_{q,2^v}$ is a $q$-translation
and the Type-I duadic splittings of $\Z_{2^v}$
given by $\tau_{2^u}$ exist.
\end{theorem}

\begin{proof}
Obviously, any one of the three statements implies that $v\ge 1$.

(i) $\Leftrightarrow$ (ii).~ This follows from Lemma \ref{l:t,duadic} immediately.

(ii) $\Rightarrow$ (iii).~
By Lemma \ref{l:s,duadic}, (ii) implies that
\begin{equation*}
 \nu_2(q^j-1)+\nu_2(q-1)> v, ~~~~~~ \forall~ j\ge 0.
\end{equation*}
Taking $j=1$,
we get $\nu_2(q-1)+\nu_2(q-1)> v$; i.e., (iii) holds.

(iii) $\Rightarrow$ (ii).~
Assume that (iii) holds, i.e., $0<v$ and $v-\nu_2(q-1)<\nu_2(q-1)$.
Since $q$ is odd, $0<\nu_2(q-1)$ and $v-\nu_2(q-1)<v$.
Hence
$$
 0\le\max\{0,~v-\nu_2(q-1)\} <\min\{v,~\nu_2(q-1)\}.
$$
Thus the integer $u$ satisfying Eqn~\eqref{e:1,t} exists.
Set $t=2^u$, then $\nu_2(t)=u$.
Because $v-\nu_2(q-1)\le u=\nu_2(t)$, by Remark~\ref{nu} (iii),
$\tau_t=\P_{1,t}\in G_{q,2^v}$ is a $q$-translation.
Since Eqn \eqref{e:1,t} implies that $\nu_2(t)<\nu_2(q-1)$,
by Corollary~\ref{c:t,duadic},
Type-I duadic splittings of $\Z_{2^v}$ given by $\tau_t$ exist.
\end{proof}

The following lemma is a bridge from the case where $n=2^v$ to
the general case where $n$ is any positive integer coprime to $q$.

\begin{lemma}\label{reduce to 2}
Let $v=\nu_2(n)$ and $n=2^{v}n'$ with $2\nmid n'$,
let $\P_{s,t}\in G_{q,n}$.
The following two statements are equivalent to each other:

(i)~ The Type-I duadic splittings of $\Z_n$ given by $\P_{s,t}$ exist.

(ii)~ The Type-I duadic splittings of $\Z_{2^{v}}$ given by $\P_{s,t}$ exist.
\end{lemma}

\begin{proof} By Chinese Remainder Theorem, we have a natural isomorphism
\begin{equation}
 \Z_n\CRT \Z_{n'}\times \Z_{2^v}
\end{equation}
which maps any $k\in\Z_{n}$ to
$\big(k~({\rm mod}~n'),\,k~({\rm mod}~2^v)\big) \in\Z_{n'}\times\Z_{2^v}$.
Obviously, $s\in\Z_{n'}^*$ and $s\in\Z_{2^v}^*$;
$qt\equiv t\!\pmod{n'}$ and $qt\equiv t\!\pmod{2^v}$.
That is, we also have the $q$-permutation $\P_{s,t}$ of $\Z_{n'}$ and of $\Z_{2^v}$.
Since the cardinality of $\Z_{n'}$ is $|\Z_{n'}|=n'$ which is odd,
at least one of the $\P_{s,t}$-orbits on $\Z_{n'}/\mu_q$ has odd length.
By Remark \ref{r:G action}~(iii),
the length of any $\P_{s,t}$-orbit on $\Z_{n}/\mu_q$ is even
if and only if the length of any $\P_{s,t}$-orbit on $\Z_{2^v}/\mu_q$ is even.
By Remark \ref{r:G action} (i),
the statements (i) and (ii) are equivalent to each other.
\end{proof}

By Lemma \ref{reduce to 2}, all the previous results on $\Z_{2^v}$
also hold for the general case~$\Z_n$. In particular,
from Theorem \ref{t:duadic} we can get the following result.

\begin{theorem}\label{t:n:duadic}
Let $n=2^{\nu_2(n)}n'$ (hence $2\nmid n'$)
be any positive integer coprime to $q$.
Then the following three statements are equivalent to each other:
\begin{itemize}
\item[\rm(i)]
The Type-I duadic splittings of $\Z_{n}$ exist.
\item[\rm(ii)]
The Type-I duadic splittings of $\Z_{n}$ given by $q$-translations exist.
\item[\rm(iii)]
$0<\nu_2(n)<2\cdot\nu_2(q-1)$.
\end{itemize}
If it is the case, then there is an integer $u$ such that
\begin{equation}\label{e:n:1,t}
\max\{0,~\nu_2(n)-\nu_2(q-1)\}\le u <\min\{\nu_2(n),~\nu_2(q-1)\};
\end{equation}
hence $\P_{1,2^u n'}=\tau_{2^u n'}\in G_{q, n}$ is a $q$-translation
and the Type-I duadic splittings of $\Z_{n}$
given by $\tau_{2^u n'}$ exist.
\end{theorem}

\begin{proof} Let $v=\nu_2(n)$, i.e., $n=2^v n'$ with $2\nmid n'$.

(i) $\Rightarrow$ (iii).~
By (i) we assume that $\P_{s,t}\in G_{q,n}$
such that the Type-I duadic splittings of $\Z_n$ given by $\P_{s,t}$ exist.
By Lemma \ref{reduce to 2},
the Type-I duadic splittings of $\Z_{2^v}$ given by $\P_{s,t}$ exist.
By Theorem \ref{t:duadic}, $0<v<2\cdot\nu_2(q-1)$; i.e., (iii) holds.

(iii) $\Rightarrow$ (ii).~
By Theorem \ref{t:duadic}, (iii) implies that there is an integer $u$
which satisfies Eqn \eqref{e:n:1,t}.
It is obvious~that $q\cdot 2^un'\equiv 2^un'\!\pmod{n'}$.
By Eqn \eqref{e:n:1,t}, $v-\nu_2(q-1)\le u=\nu_2(2^un')$, i.e.,
$q\cdot 2^un'\equiv 2^un'\!\pmod {2^v}$, cf. Remark~\ref{nu}~(iii).
So $q\cdot 2^un'\equiv 2^un'\!\pmod{2^v n'}$ as $2^v$ and $n'$ are coprime.
Recalling that $n=2^v n'$, we get a $q$-translation
$\P_{1,2^v n'}=\tau_{2^un'}\in G_{q,n}$ of $\Z_n$.
By Eqn \eqref{e:n:1,t} again, $\nu_2(2^u n')=u<\nu_2(q-1)$.
By Corollary \ref{c:t,duadic}, the Type-I duadic splitiings
of $\Z_{2^v}$ given by $\tau_{2^u n'}$ exist.
By Lemma \ref{reduce to 2}, the Type-I duadic splitiings
of $\Z_{n}$ given by $\P_{1,2^v n'}=\tau_{2^u n'}$ exist.

(ii) $\Rightarrow$ (i).~ Trivial.
\end{proof}

\section{Constructions of iso-self-dual cyclic codes}

\subsection{A program of construction of iso-self-dual cyclic codes}

From Theorem \ref{t:n:duadic} and Theorem \ref{t:dual splitting},
we get immediately a necessary and sufficient condition for the existence
of iso-self-dual cyclic codes.

\begin{theorem}\label{existence iso-self-dual}
Iso-self-dual cyclic codes over $\F_q$ of length $n$ exist
if and only if $0<\nu_2(n)<2\cdot\nu_2(q-1)$.
\end{theorem}

\begin{remark}\label{a construction}\rm
More precisely, Theorem \ref{t:n:duadic}, Theorem \ref{t:dual splitting}
and their arguments provide a program to construct iso-self-dual cyclic codes
over $\F_q$ of length $n$.
\begin{itemize}
\item[S1.] (determination of existence of iso-self-dual cyclic codes)

 if $\nu_2(n)=0$ or $\nu_2(n)\ge 2\cdot\nu_2(q-1)$,
then output
``{\tt iso-self-dual cyclic codes over $\F_q$ of length $n$ do not exist}'';

otherwise take an integer $u$ satisfying that
$$
\max\{0,~\nu_2(n)-\nu_2(q-1)\}\le u <\min\{\nu_2(n),~\nu_2(q-1)\},
$$
set $n'=n/2^{\nu_2(n)}$, $t=2^u n'$ and go to S2.
\item[S2.]
(construction of iso-self-dual cyclic codes)
\begin{itemize}
\item[1)]
compute $q$-cosets on $\Z_n$ to get the quotient set $\Z_n/\mu_q$;
\item[2)]
compute the $\tau_{t}$-orbits on $\Z_n/\mu_q$;
\item[3)]
for every $\tau_{t}$-orbit $(Q_1,Q_2,\cdots,Q_{2m})$
on $\Z_n/\mu_q$
(the length of any $\tau_{t}$-orbit must be even),
place  $Q_1,Q_3,\cdots,Q_{2m-1}$ into the subset $P$
(cf. Remark \ref{r:G action} (i));
\item[4)]
output the cyclic code $C_P\subseteq R_n$ with check polynomial $f_P(X)$
and the dual code $C_P^\bot=\T_{-1,t}(C_P)=C_{-\bar P}$.
\end{itemize}
\end{itemize}
\end{remark}

\begin{example}\rm
Let $q=3$, $n=8$.
Then $\nu_2(n)=3>2=2\nu_2(q-1)$;
so there are no iso-self-dual cyclic codes over $\F_3$ of length $8$.
In fact, in this case,
$$\Z_8/\mu_3=\big\{\{0\},~ \{4\},~
    \{1,3\},~  \{2,6\},~ \{5,7\}\big\}
$$
which consists of $5$ $q$-cosets; any permutation on $\Z_8/\mu_3$
has at least one orbit of odd length.
\end{example}

\begin{example}\rm
Let $q=5$, $n=8$.
Then $\nu_2(n)=3<4=2\nu_2(q-1)$; the condition
in Theorem \ref{existence iso-self-dual} is satisfied.
We can take $u=1$:
$$
\max\{0,~\nu_2(n)-\nu_2(q-1)\}=1\le 1<2=\min\{\nu_2(n),~\nu_2(q-1)\}.
$$
Then $t=2$. The set of $q$-cosets
$$\Z_8/\mu_5=\big\{\{0\},~ \{4\},~ \{2\},~ \{6\},~
    \{1,5\},~  \{3,7\}\big\} .
$$
The orbits of the $q$-translation $\tau_t=\tau_2$ are as follows:
$$
\big(\{0\},~\{2\},~\{4\},~\{6\}\big),~~~~
\big(\{1,5\},~\{3,7\}\big).
$$
We put $\{0\},~\{4\},~\{1,5\}$ into $P$, and get $P=\{0,1,4,5\}$.
Then $\P_{1,t}(P)=\bar P$ is the complement of $P$ in $\Z_n$.
Take $\theta$ to be a primitive $8$-th root of unity such that $\theta^2=2$.
Then $C_P$ is an iso-self-dual cyclic code over $\F_5$ of length $8$
with the check polynomial $f_P(X)$ as follows (note that $\theta^5=-\theta$):
\begin{eqnarray*}
 f_P(X)&=&(X-\theta^0)(X-\theta^4)(X-\theta)(X-\theta^5)\\
 &=& (X-1)(X+1)(X^2-2)\\
 &=& X^4+2X^2+2.
\end{eqnarray*}
And $C_P^\bot=\T_{-1,t}(C_P)$.
Thus $f_{\P_{-1,t}(P)}(X)$ is a generator polynomial of $C_P$.
By Corollary \ref{c:T1 P}, $f_{\P_{-1,t}(P)}(X)$ can be obtained
from $f_P(X)$ as follows (note that $\theta^{-t}=2^{-1}=3$):
$$
 f_{\P_{1,t}(P)}(X)=\theta^{t|P|}f_P(\theta^{-t}X)
 =X^4-2X^2+2;
$$
and
$$
f_{\P_{-1,t}(P)}(X)=f_{\P_{-1,0}(\P_{1,t}(P))}(X)
=2^{-1}(1-2X^2+2X^4)=X^4-X^2-2.
$$
\end{example}

\subsection{The case where $\nu_2(n)=1$}

For the length $n$, there is a specific case where
the iso-self-dual cyclic codes of length $n$ always exist.

\begin{theorem}\label{t:n=2n'}
Assume that $n$ is an even positive integer such that $n'=n/2$ is odd.
Then, for any odd prime power $q$ coprime to $n$,
the length of any $\tau_{n'}$-orbit
on the quotient set $\Z_n/\mu_q$ of $q$-cosets equals to $2$; in particular,
Type I duadic splittings of $\Z_n$ given by $\tau_{n'}$ exist.
\end{theorem}

\begin{proof}
For any $i\in\Z_n$, $\tau_{n'}^2(i)=i+2n'\equiv i\!\pmod n$.
That is, the order of the permutation $\tau_{n'}$ equals to $2$;
Then the length of any $\tau_{n'}$-orbit on $\Z_n/\mu_q$
equals to either $2$ or $1$.
Suppose a $\tau_{n'}$-orbit on $\Z_n/\mu_q$ has length $1$;
i.e., consists of only one $q$-coset $Q$;
then $\tau_{n'}(Q)=Q$ and by Lemma \ref{l:Q fixed}
there is a $k\in Q$ and an integer $j\ge 0$ such that
$(q^j-1)k\equiv n'\!\pmod{n}$; this is impossible because
both $q^j-1$ and $n$ are even but $n'$ is odd.
Therefore, any $\tau_{n'}$-orbit on $\Z_{n}/\mu_q$ has length $2$.
By Remark \ref{a construction}, taking any one $q$-coset
from each $\tau_{n'}$-orbit on $\Z_n/\mu_q$ to put into $P$,
we get a partition $\Z_n=P\bigcup\tau_{n'}(P)$.
\end{proof}

\begin{remark}\label{poly S2}\rm
Let notations be as in Theorem \ref{t:n=2n'}.
In that case we can omit the step S1 in Remark \ref{a construction};
moreover, we have a corresponding polynomial version of the
step S2 in Remark \ref{a construction}. To state it, we first
describe two operations on polynomials, which correspond to
the $q$-permutations $\tau_{n'}$ and $\mu_{-1}$ respectively.

Let $a(X)=X^d +a_{d-1}X^{d-1}+\cdots+a_{0}$
be a monic polynomial of degree $\deg a(X)=d$
with non-zero constant $a_0$.
\begin{itemize}
\item
$\widetilde a(X)=(-1)^{\deg a(X)}a(-X)=X^d-a_{d-1}X^{d-1}+\cdots+(-1)^da_0$,
called the {\em alternating polynomial} of $a(X)$;
\item
$a^*(X)=a_0^{-1}X^{\deg a(X)}a(X^{-1})$,
called the {\em monic reciprocal polynomial} of $a(X)$;
\end{itemize}

For a primitive $n$-th root $\theta$ of unity,
$\theta^{-n'}=-1$. Then, by Corollary~\ref{c:T1 P},
the program S2 in Remark \ref{a construction} can be translated
as follows:

\medskip
S2. (construction of iso-self-dual cyclic codes)
\begin{itemize} \leftskip2em
\item[1)]
decompose $X^n-1$ into a product of monic irreducible factors;
\item[2)]
partition the monic irreducible factors of $X^n-1$ into
pairs such that in every pair
one polynomial is alternating of the other one;
\item[3)]
for every pair of the partition in 2), select one factor from the pair;
\item[4)]
set $h(X)$ to be the product of the factors which are selected in 3);
output the cyclic code $C\subseteq R_n$ with check polynomial $h(X)$,
and the dual code $C^\bot$ with check polynomial ${\widetilde h}^*(X)$.
\end{itemize}
\end{remark}

We illustrate the above construction with an example.

\begin{example}\rm
Let $n=10$, $q=3$. By the above theorem the Type-I duadic splittings
of $\Z_{10}$ given by the $3$-translation $\tau_{5}$ exist. By the program
in Remark~\ref{a construction}, we can get one of them as follows.
The set of $q$-cosets
\begin{equation}\label{quotient set}
 \Z_{10}/\mu_3=\big\{\{0\},~ \{5\},~\{1,3,9,7\},~ \{2,6,8,4\}\big\}.
\end{equation}
The $\tau_{5}$-orbits on $\Z_{10}/\mu_3$ are:
\begin{equation}
 \big(\{0\},~ \{5\}\big),~~~~ \big(\{1,3,9,7\},~ \{2,6,8,4\}\big).
\end{equation}
For every $\tau_5$-orbit, select one $q$-coset as follows:
\begin{equation}
   \{0\},~~~~ \{1,3,9,7\}.
\end{equation}
Place them into the subset $P$ to get
\begin{equation}
 P=\{0,~1,~3,~9=-1,~7=-3\},
\end{equation}
and $\tau_5(P)=\{2,4,5,6,8\}=\bar P$,
\begin{equation}\label{bar P}
 \mu_{-1}\tau_5(P)=-\bar P=\{-2,-4,-5,-6,-8\}=\bar P.
\end{equation}

Corresponding to the equations \eqref{quotient set}--\eqref{bar P},
we list the polynomial versions described in Remark \ref{poly S2} as follows.

The monic irreducible decomposition
$$
 X^n-1=(X-1)(X+1)(X^4+X^3+X^2+X+1)(X^4-X^3+X^2-X+1).
 \eqno(4.1')
$$
The pairwise partition of the irreducible factors
such that in every pair one polynomial
is the alternating polynomial of the other one are:
$$
 \{X\!-\!1,~X\!+\!1\}, ~~\{X^4-X^3+X^2-X+1,~ X^4+X^3+X^2+X+1\}.
 \eqno(4.2')
$$
Select one from each pair as follows:
$$
  X-1, ~~~ X^4-X^3+X^2-X+1.
 \eqno(4.3')
$$
The check polynomial of $C$ is:
$$
 h(X)=(X-1)(X^4-X^3+X^2-X+1)=X^5-2X^4+2X^3-2X^2+2X-1.
 \eqno(4.4')
$$
and $\widetilde h(X)=X^5+2X^4+2X^3+2X^2+2X+1$,
the check polynomial of $C^\bot$ is:
$$
 \widetilde h^*(X)=X^5+2X^4+2X^3+2X^2+2X+1.
 \eqno(4.5')
$$
\end{example}

\subsection{A class of iso-self-dual MDS cyclic codes}

From Theorem \ref{t:n=2n'}, we can further construct a class
of iso-self-dual MDS cyclic codes. Note that
$q\equiv 1\!\pmod 4$ if and only if 
$\nu_2(q+1)=1$.

\begin{theorem}\label{t:MDS}
Assume that $q\equiv 1\!\pmod 4$, and $n=q+1$. Let $n'={n}/{2}$ and
$$P=\Big\{-\frac{n-2}{4},~
    ~\cdots,~ -1,~ 0,~ 1,~\cdots,~ \frac{n-2}{4}\Big\}~\subseteq~ {\Bbb Z}_n.$$
Then $P$ is a $\mu_q$-invariant subset of $\Z_n$  and
\begin{itemize}
\item[(i)]
$C_P\subseteq R_n$ 
is an iso-self-dual cyclic code, and $\T_{-1,n'}(C_P)=C_P^\bot$;
\item[(ii)]
$C_P$ is an alternant code over ${\F}_q$ of length $n$ by restricting
a generalized Reed-Solomon code over ${\F}_{q^2}$ of length $n$
to the subfield $\F_q$,
in particular, $C_P$ is a $[q+1,\,\frac{q+1}{2},\,\frac{q+3}{2} ]$ MDS code.
\end{itemize}
\end{theorem}

\begin{proof} Note that $|P|=\frac{n}{2}=n'$ is odd,
$$
\bar P=\Z_{n}\backslash P=\Big\{\frac{n+2}{4},~
\frac{n+2}{4}+1,~\cdots,~\frac{3n-2}{4}\Big\},
$$
and $q\equiv -1\!\pmod n$.
Then for any $i\in P$, i.e., $-\frac{n-2}{4}\le i\le\frac{n-2}{4}$,
we have $qi\equiv -i\!\pmod n$, hence
$$
 \frac{n-2}{4}\ge -i\ge -\frac{n-2}{4};
$$
i.e., $qi\in P$. We get that $P$ is $\mu_q$-invariant.

(i).~ Let $i\in P$, i.e., $-\frac{n-2}{4}\le i\le\frac{n-2}{4}$.
Note that $\tau_{n'}(i)=i+\frac{n}{2}\!\pmod n$,
and $-\frac{n-2}{4}\equiv n-\frac{n-2}{4}=\frac{3n+2}{4}\!\pmod n$.
Then
$$
 \tau_{n'}(i)=i+\frac{n}{2}\ge \frac{n}{2}-\frac{n-2}{4}\ge
  \frac{n+2}{4}>\frac{n-2}{4};
$$
and
$$
 \tau_{n'}(i)=i+\frac{n}{2}
 \le \frac{n}{2}+\frac{n-2}{4}=\frac{3n-2}{4}<\frac{3n+2}{4}.
$$
We obtain that $\tau_{n'}(i)\notin P$.
So $P\bigcap\tau_{n'}(P)=\emptyset$;
and $\Z_n=P\bigcup\tau_{n'}(P)$ is a partition.
By Theorem \ref{t:dual splitting}, (i) holds.

(ii).~ Let $\theta$ be a primitive $n$-th root of unity. Since
$n=q+1$ divides $q^2-1$, we see that $\theta\in\F_{q^2}$.
The subset $P$ is of course $\mu_{q^2}$- invariant,
By $\tilde C_P$ we denote the cyclic code over $\F_{q^2}$ of length $n$
(i.e., $\tilde C_P\subseteq \F_{q^2}[X]/\langle X^n-1\rangle$)
with check polynomial
$$f_P(X)=\prod_{i\in P}(X-\theta^i) \in \F_{q^2}[X],$$
hence $f_{\bar P}(X)=\prod_{h\in\bar P}(X-\theta^h)$
is a generator polynomial of $\tilde C_P$.
Of course, $\tilde C_P$ is iso-self-dual.
Denote $\ell=\frac{n-2}{4}$. We claim that
\begin{equation}\label{e:RS}
\tilde C_p\!=\!\big\{\big(a(1), \theta^{\ell}a(\theta^{-1}),
  \cdots,\theta^{(n-1)\ell}a(\theta^{-(n-1)})\big)
   \,\big|\,a(X)\in\F_{q^2}[X], \deg a(X)<n'\big\}
\end{equation}
where the right hand side is a generalized Reed-Solomon code
over $\F_{q^2}$ of length~$n$.

Note that $\{\theta^h\mid h\in{\bar P}\}$ is the defining set of
the cyclic code $\tilde C_P$.  
Any word $\big(a(1), \theta^{\ell}a(\theta^{-1}),
 \cdots, \theta^{(n-1)\ell}a(\theta^{-(n-1)})\big)$,
for $a(X)\in\F_{q^2}[X]$ with $\deg a(X)<n'$, corresponds to
the polynomial
$$C_a(X)=\sum_{i=0}^{n-1}\theta^{\ell i} a(\theta^{-i})X^i\in \F_{q^2}[X].$$
To proof Eqn \eqref{e:RS}, it is enough to show that,
for any $h\in\bar P$ and any
$a(X)\in\F_{q^2}[X]$ with $\deg a(X)<n'$,
the $\theta^h$ is a root of $C_a(X)$.
Write $a(X)=\sum_{j=0}^{n'-1}a_jX^j$. Then
\begin{equation*}
C_a(\theta^h)=\sum_{i=0}^{n-1} \theta^{\ell i} a(\theta^{-i})\theta^{h i}
=\sum_{i=0}^{n-1} \sum_{j=0}^{n'-1} a_j\theta^{-j i} \theta^{\ell i} \theta^{h i};
\end{equation*}
i.e.,
\begin{equation}\label{zero C_a}
C_a(\theta^h)=\sum_{j=0}^{n'-1} a_j\sum_{i=0}^{n-1}  \theta^{(h-(j-\ell)) i}.
\end{equation}
Recall that $\ell=\frac{n-2}{4}$. For $0\le j\le n'-1=\frac{n-2}{2}$, we get
$$
  -\frac{n-2}{4}\le j-\ell \le \frac{n-2}{2}-\frac{n-2}{4}=\frac{n-2}{4};
$$
in other words, $j-\ell\in P$. However, $h\in\bar P$.
So $j-\ell~{\not\equiv}~h\!\pmod n$, equivalently,
$h-(j-\ell)~{\not\equiv}~0\!\pmod n$.
Noting that $\theta$ is a primitive $n$-th root of unity, we see that
$\theta^{h-(j-\ell)}\neq 1$ and $\theta^n=1$. Thus,
$$
\sum_{i=0}^{n-1}  \theta^{(h-(j-\ell)) i}
=\frac{\theta^{(h-(j-\ell))n}-1}{\theta^{h-(j-\ell)}-1}=0.
$$
Return to Eqn \eqref{zero C_a}, we obtain that $C_a(\theta^h)=0$
for all $h\in\bar P$. The conclusion of Eqn \eqref{e:RS} is proved.

Finally, since the subset $P$ of $\Z_n$
is $\mu_q$-invariant, the polynomial $f_P(X)$ is in fact
an $\F_q$-polynomial. The cyclic code $C_P$ over $\F_q$ has $f_P(X)$
as a check polynomial, so $C_P$ consists of the codewords
of $\tilde C_P$ whose coefficients are all in~$\F_q$.
We are done.
\end{proof}

\medskip
The performance of the iso-self-dual MDS cyclic codes 
in Theorem~\ref{t:MDS} would be nice.
These codes have the following properties:
\begin{itemize}
\item
they have good algebraic structure: they are cyclic and iso-self-dual;
\item
they have good weight structure: they are MDS and formally self-dual;
\item
they can be systematically encoded by feed-back shift registers;
\item
they can be decoded by Berlekamp-Welch decoding algorithm.
\end{itemize}

The following is an instance of the class of codes
constructed in Theorem~\ref{t:MDS}.

\begin{example}\rm
Let $q=5$, $n=q+1=6$.
Then $q=5\equiv 1\!\pmod 4$, $\nu_2(n)=1$ and $n'=6/2=3$.
By Theorem \ref{t:MDS}, we take $t=n'=3$, and
$$P=\{0,1,-1=5\}, \mbox{ ~~ hence ~~ }
 \bar P=\Z_{10}\backslash P=\{2,3,4\} $$
and $\tau_{t}(P)=\bar P$.
Take $\theta$ to be a primitive $6$-th root of unity such that $\theta^3=-1$.
Then $C_P$ is an iso-self-dual MDS cyclic code over $\F_5$ of length $6$
with the check polynomial $f_P(X)$
and the generator polynomial $f_{\bar P}(X)$
as follows (note that $\theta^3=-1$):
\begin{eqnarray*}
f_P(X)&=&(X-\theta^0)(X-\theta)(X-\theta^5)=(X-1)(X^2-X+1)\\
 &=&X^3-2X^2+2X-1,\\
f_{\bar P}(X)&=&(X-\theta^3)(X-\theta^2)(X-\theta^4)=(X+1)(X^2+X+1)\\
 &=&X^3+2X^2+2X+1.
\end{eqnarray*}
And $C_P^\bot=\T_{-1,t}(C_P)$ has a check polynomial
$f_{-\bar P}(X)=X^3+2X^2+2X+1$.
By Theorem \ref{t:MDS}(ii) and Eqn \eqref{e:RS},
$C_P$ is the alternant code by restricting the following
generalized Reed-Solomon code over $\F_{25}$ of length $6$:
$$
\big\{\big(a(1), \theta a(\theta^{-1}),
  \cdots,\theta^{5}a(\theta^{-5})\big)
   \;\big|\;a(X)\in\F_{25}[X],\; \deg a(X)<3\big\}.
$$

\end{example}

\section*{Acknowledgements}
The research of the authors is supported
by NSFC with grant number 11271005.


\begin{thebibliography}{99}

\bibitem{AB} J. L. Alperin, R. B. Bell,
Groups and Representations, GTM 162, Springer-Verlag, New York, 1997.

\bibitem{AKS}
S. A. Aly, A. Klappenecker, P. K.  Sarvepalli, Duadic group algebra codes,
In Proc. Int. Symp. Inf. Theory, Adelaide, Australia, (2007),  2096-2100.


\bibitem{B08} T. Blackford,
Negacyclic duadic codes,  Finite Fields Appl.,
{\bf 14}(2008),  930-943.

\bibitem{B13} T. Blackford,
Isodual constacyclic codes,
Finite Fields Appl., {\bf 24}(2013), 29-44.

\bibitem{BP}
R. A. Brualdi, V. Pless, Polyadic codes,
Discr. Appl. Math., {\bf 25}(1989), 3-17.

\bibitem{CD}
B. Chen and H. Q. Dinh, A note on isodual constacyclic codes,
Finite Fields Appl., vol. 29, pp. 243-246, Sep. 2014.

\bibitem{CDFL}
B. Chen, H. Q. Dinh, Y. Fan, S. Ling,
Polyadic constacyclic codes, 
IEEE Trans. Inform. Theory {\bf 61}(2015), no.9, 4895-4904.


\bibitem{CFLL} B. Chen, Y. Fan, L. Lin, H. Liu,
Constacyclic codes over finite fields,
Finite Fields Appl., {\bf 18}(2012), 1217-1231.



\bibitem{DLX}
C. Ding, K.Y. Lam, C. Xing,
Enumeration and construction of all duadic codes of length $p^m$,
Fund. Inform. {\bf 38}(1999),  149-161.

\bibitem{DL} H. Q. Dinh, Lopez-Permouth,
Cyclic and negacyclic codes over finite chain rings,
IEEE Trans. Inform. Theory {\bf 50}(2004), no.8, 1728-1744.

\bibitem{DP}
C. Ding, V. Pless, Cyclotomy and duadic codes of prime lengths,
IEEE Trans. Inform. Theory, {\bf 45}(1999), 453-466.

\bibitem{FZ12} Y. Fan, G. Zhang,
On the existence of self-dual permutation codes of finite groups,
Des. Codes Cryptogr., {\bf 62}(2012),  19-29.

\bibitem{FZ16} Y. Fan, L. Zhang,
Iso-orthogonality and Type-II duadic constacyclic codes,
Finite Fields Appl. {\bf 41} (2016) 1-23.

\bibitem{FZ16-a} Y. Fan, L. Zhang,
Galois self-dual constacyclic codes,
Des. Codes Cryptogr., DOI 10.1007/s10623-016-0282-8,
online: 29 Sept. 2016.

\bibitem{HK}
S. Han, J-L. Kim, Computational results of duadic double circulant codes,
J.  Appl.  Math. Comput.,  {\bf 40}(2012),  33-43.

\bibitem{HP} W.C. Huffman, V. Pless,
Fundamentals of Error-Correcting Codes,
Cambridge University Press, Cambridge, 2003.

\bibitem{JLS}
S. Jitman, S. Ling, P. Sol\'e,
Hermitian self-dual Abelian codes,
IEEE Trans. Inform. Theory, 60(2014),1496-1507.

\bibitem{LMP}
J. S. Leon, J. M. Masley,  V. Pless, Duadic codes,
IEEE Trans. Inform. Theory,  {\bf 30}(1984),
709-714.

\bibitem{LX}
S. Ling, C. Xing, Polyadic codes revisited,
IEEE Trans. Inform. Theory, 50(2004),200-207.

\bibitem{MW}
C. Martinez-P\'erez, W. Willems,
Self-dual codes and modules for finite groups in characteristic two,
IEEE Trans. Inform. Theory 50(2004), no. 8, 1798-1803.

\bibitem{P}
V. Pless, Duadic codes revisited,  Congressus Numeratium, {\bf 59}(1987),
225-233.

\bibitem{R}
J. J. Rushanan,   Duadic codes and difference sets,
J. Combin. Theory Set. A,  {\bf 57}(1991), 254-61.

\bibitem{SBDR}
A. Sharma, G.K. Bakshi, V.C. Dumir, M. Raka,
Cyclotomic numbers and primitive idempotents
in the ring $GF(q)[X]/(X^{p^n}-1)$,
Finite Fields Appl. 10 (4) (2004) 653-673.

\bibitem{S}
M. H. M. Smid, Duadic codes, IEEE Trans. Inform. Theory,
 {\bf 33}(1987),  432-433.

\bibitem{WZ}
H. N. Ward, L. Zhu,  Existence of abelian group codes partitions,
J. Combin. Theory Ser. A,  {\bf67}(1994), 276-281.

\bibitem{W}
W. Willems, A note on self-dual group codes, IEEE Trans. Inform. Theory 48(2002), 3107-3109.

\end{thebibliography}
\end{document}